\documentclass[11pt]{article}
\usepackage{microtype}
\usepackage{fullpage,amsthm,amssymb,amsmath}
\usepackage{graphicx,algorithmic,algorithm}
\usepackage{enumerate}
\usepackage{tikz}
\usetikzlibrary{shapes}
\usepackage{hyperref}
\usepackage{complexity}
\usepackage[normalem]{ulem}
\usepackage{caption, subcaption}
\usepackage{todonotes}

\usepackage{xcolor}
\definecolor{dark-red}{rgb}{0.4,0.15,0.15}
\definecolor{dark-blue}{rgb}{0.15,0.15,0.4}
\definecolor{medium-blue}{rgb}{0,0,0.5}
\definecolor{gray}{rgb}{0.5,0.5,0.5}

\hypersetup{
    pdftitle=   {A note on the complexity of Feedback Vertex Set parameterized by mim-width},
   pdfauthor=  {Lars Jaffke, O-joung Kwon and Jan Arne Telle}, 
    colorlinks, linkcolor={dark-red},
    citecolor={dark-blue}, urlcolor={medium-blue}
}

\usepackage{authblk}
\newtheorem{theorem}{Theorem}

\newtheorem{proposition}[theorem]{Proposition}
\newtheorem{corollary}[theorem]{Corollary}

\newtheorem*{thm:main}{Theorem~\ref{thm:main}}

\newcommand\abs[1]{\lvert #1\rvert}
\newcommand\card[1]{\lvert #1 \rvert}
\newcommand\cardd[1]{\vert\lvert #1 \rvert\rvert}

\newcommand\mimw{\operatorname{mimw}}
\newcommand\mimval{\operatorname{mim}}

\newcommand\cM{\mathcal{M}}

\newcommand\cS{\mathcal{S}}

\usepackage{xspace}

\newcommand*{\defeq}{\mathrel{\vcenter{\baselineskip0.5ex \lineskiplimit0pt
                     \hbox{\scriptsize.}\hbox{\scriptsize.}}}%
                     =}

\newcommand{\parproblemdef}[4]
{
\begin{quote}
\textsc{#1}\\
\textbf{Input:} #2\\
\textbf{Parameter:} #3\\
\textbf{Question:} #4
\end{quote}
}

\theoremstyle{definition}

\theoremstyle{remark}

\newtheorem{claim}[theorem]{Claim}
\newtheorem{observation}[theorem]{Observation}

\newtheorem*{remark*}{Remark}
\newtheorem*{question*}{Open Problem}

\newenvironment{clproof}{\begin{proof}\renewcommand{\qedsymbol}{\claimqed}}{\end{proof}\renewcommand{\qedsymbol}{\plainqed}}

\let\plainqed\qedsymbol

\newcommand\dectree{T}
\newcommand\decf{\mathcal{L}}

\newcommand{\specialindex}[1]{#1+\epsilon}

\begin{document}

\title{A note on the complexity of Feedback Vertex Set parameterized by mim-width}

\author[1]{Lars Jaffke\thanks{Supported by the Bergen Research Foundation (BFS).}}

\author[2]{O-joung Kwon\thanks{Supported by the European Research Council (ERC) under the European Union's Horizon 2020 research and innovation programme (ERC consolidator grant DISTRUCT, agreement No. 648527).}}

\author[1]{Jan Arne Telle}

\affil[1]{Department of Informatics, University of Bergen, Norway. \protect\\ \texttt{\{lars.jaffke, jan.arne.telle\}@uib.no}}
\affil[2]{Logic and Semantics, Technische Universit\"at Berlin, Berlin, Germany. \protect\\ \texttt{ojoungkwon@gmail.com}}

\date\today

\maketitle

\begin{abstract}
	We complement the recent algorithmic result that {\sc Feedback Vertex Set} is $\XP$-time solvable parameterized by the mim-width of a given branch decomposition of the input graph~\cite{JaffkeKT2017b} by showing that the problem is $\W[1]$-hard in this parameterization. The hardness holds even for linear mim-width, as well as for $H$-graphs, where the parameter is the number of edges in $H$.
	To obtain this result, we adapt a reduction due to Fomin, Golovach and Raymond~\cite{FGR17}, following the same line of reasoning but adding a new gadget.
\end{abstract}

\section{Preliminaries}
In this note (which will later be merged with the companion paper~\cite{JaffkeKT2017b}), unless stated otherwise, a graph $G$ with vertex set $V(G)$ and edge set $E(G) \subseteq \binom{V(G)}{2}$ is finite, undirected, simple and connected. We let $\card{G} \defeq \card{V(G)}$ and $\cardd{G} \defeq \card{E(G)}$. For an integer $k > 0$, we let $[k] \defeq \{1,\ldots,k\}$.

For a vertex $v \in V(G)$, we denote by $N(v)$ the set of {\em neighbors} of $v$, i.e.~$N(v) \defeq \{w \mid vw \in E(G)\}$. 

For two graphs $G$ and $H$ we denote by $H \subseteq G$ that $H$ is a {\em subgraph} of $G$ i.e.\ that $V(H) \subseteq V(G)$ and $E(H) \subseteq E(G)$. For a vertex set $X \subseteq V(G)$, we denote by $G[X]$ the subgraph of $G$ {\em induced by $X$} i.e.~$G[X] \defeq (X, E(G) \cap \binom{X}{2})$.
For two (disjoint) vertex sets $X, Y \subseteq V(G)$, 
	we denote by $G[X, Y]$ the bipartite subgraph of $G$ with bipartition $(X, Y)$ such that for $x\in X, y\in Y$, $x$ and $y$ are adjacent in $G$ if and only if they are adjacent in $G[X,Y]$. A {\em cut} of $G$ is a bipartition $(A, B)$ of its vertex set.
	A set $M$ of edges is a \emph{matching} if no two edges in $M$ share an endpoint, and a matching $\{a_1b_1, \ldots, a_kb_k\}$ is  \emph{induced} if there are no other edges in the subgraph induced by 
	$\{a_1, b_1, \ldots, a_k, b_k\}$.
	
	Let $uv \in E(G)$. We call the operation of adding a new vertex $x$ to $V(G)$ and replacing $uv$ by the path $uxv$ the {\em edge subdivision} of $uv$. We call a graph $G'$ a {\em subdivision of $G$} if it can be obtained from $G$ by a series of edge subdivisions.

\smallskip
\noindent{\bf Mim-width.} For a graph $G$ and a vertex subset $A$ of $G$, we define $\mimval_G(A)$ to be the maximum size of an induced matching in $G[A, V(G) \setminus A]$. 

A pair $(\dectree, \decf)$ of a subcubic tree $\dectree$ and a bijection $\decf$ from $V(G)$ to the set of leaves of $\dectree$ is called a \emph{branch decomposition}. If $\dectree$ is a caterpillar, then $(\dectree, \decf)$ is called a {\em linear branch decomposition}.
For each edge $e$ of $\dectree$, 
let $\dectree^e_1$ and $\dectree^e_2$ be the two connected components of $\dectree-e$, and 
let $(A^e_1, A^e_2)$ be the vertex bipartition of $G$ such that for each $i\in \{1,2\}$, 
$A^e_i$ is the set of all vertices in $G$ mapped to leaves contained in $\dectree^e_i$ by $\decf$. 
The {\em mim-width of $(\dectree, \decf)$}, denoted by $\mimw(\dectree, \decf)$, is defined as $\max_{e \in E(\dectree)} \mimval_G(A^e_1)$.
The minimum mim-width over all branch decompositions of $G$ is called the {\em mim-width of $G$}. We define the {\em linear mim-width} accordingly, additionally requiring the corresponding branch decomposition to be linear.
If $\abs{V(G)}\le 1$, then $G$ does not admit a branch decomposition, and the mim-width of $G$ is defined to be $0$.

\smallskip
\noindent{\bf $H$-Graphs.} Let $X$ be a set and $\cS$ a family of subsets of $X$. The {\em intersection graph} of $\cS$ is a graph with vertex set $\cS$ such that $S, T \in \cS$ are adjacent if and only if $S \cap T \neq \emptyset$.
Let $H$ be a (multi-) graph. We say that $G$ is an {\em $H$-graph} if there is a subdivision $H'$ of $H$ and a family of subsets $\cM \defeq \{M_v\}_{v \in V(G)}$ (called an {\em $H$-representation}) of $V(H')$ where $H'[M_v]$ is connected for all $v \in V(G)$, such that $G$ is isomorphic to the intersection graph of $\cM$. 

\section{The Proof}

Very recently, Fomin et al.~\cite{FGR17} showed that $H$-graphs have linear mim-width at most $2\cdot\cardd{H}$ (Theorem 1) and that {\sc Independent Set} is $\W[1]$-hard parameterized by $k + \cardd{H}$, where $k$ denotes the solution size (Theorem 6). This implies that {\sc Independent Set} is $\W[1]$-hard for the combined parameter solution size plus linear mim-width. We will modify their reduction to show that {\sc Maximum Induced Forest} parameterized by the mim-width of a given linear branch decomposition plus the solution size remains~$\W[1]$-hard. We formally define this parameterized problem below.

\parproblemdef
	{Maximum Induced Forest/Linear Mim-Width+$k$}
	{A graph $G$, a linear branch decomposition $(\dectree, \decf)$ of $G$ and an integer $k$.}
	{$w + k$, where $w \defeq \mimw(\dectree,\decf)$.}
	{Does $G$ contain an induced forest on $k$ vertices?}

The reduction is from {\sc Multicolored Clique} where given a graph $G$ and a partition $V_1, \ldots, V_k$ of $V(G)$, the question is whether $G$ contains a clique of size $k$ using precisely one vertex from each $V_i$ ($i \in \{1,\ldots, k\}$). This problem is known to be $\W[1]$-complete~\cite{FHRV09,Pie03}.

\begin{theorem}\label{thm:main}
	{\sc Maximum Induced Forest} is $\W[1]$-hard when parameterized by $w + k$ and the hardness holds even when a linear branch decomposition of mim-width $w$ is given.
\end{theorem}
\begin{proof}
	Let $(G, V_1, \ldots, V_k)$ be an instance of {\sc Multicolored Clique}. We can assume that $k \ge 2$ and that $\card{V_i} = p$ for $i \in [k]$. If the second assumption does not hold, let $p \defeq \max_{i \in [k]} \card{V_i}$ and add $p - \card{V_i}$ isolated vertices to $V_i$, for each $i \in [k]$; we denote by $v^i_1, \ldots, v^i_p$ the vertices of $V_i$.
	
	We first obtain an $H$-graph $G''$ from an adapted version of the construction due to Fomin et al.~\cite[Proof of Theorem 6]{FGR17} as follows. The graph $H$ remains the same and is constructed as follows.
	\begin{enumerate}[1.]
		\item Construct $k$ nodes $u_1, \ldots, u_k$.\label{construction:H:vertices}
		\item For every $1 \le i < j \le k$, construct a node $w_{i, j}$ and two pairs of parallel edges $u_i w_{i, j}$ and $u_j w_{i, j}$.\label{construction:H:edges}
	\end{enumerate}
	Note that $\card{H} = k + \binom{k}{2} = k(k+1)/2$  and $\cardd{H} = 4\cdot \binom{k}{2} = 2k(k-1)$. We then construct the subdivision $H'$ of $H$ by first subdividing each edge $p$ times.
	We denote the subdivision nodes for $4$ edges of $H$ constructed for each pair $1 \le i < j \le k$ in Step~\ref{construction:H:edges} by 
$x_1^{(i, j)}, \ldots, x_p^{(i, j)}$, $y_1^{(i, j)}, \ldots, y_p^{(i, j)}$, $x_1^{(j, i)}, \ldots, x_p^{(j, i)}$, and $y_1^{(j, i)}, \ldots, y_p^{(j, i)}$.
	To simplify notation, we assume that $u_i = x_0^{(i ,j)} = y_0^{(i, j)}$, $u_j = x_0^{(j, i)} = y_0^{(j, i)}$ and $w_{i, j} = x^{(i, j)}_{p+1} = y^{(i, j)}_{p + 1} = x^{(j, i)}_{p+1} = y^{(j, i)}_{p + 1}$. 
	
	Furthermore,\footnote{To clarify, we would like to remark that this step (and everything revolving around the resulting vertices) did not appear in the reduction of Fomin et al.~\cite{FGR17} and is vital to make it work for {\sc Maximum Induced Forest}.} for $i \in [k-1]$, we subdivide the edges $x_0^{(i, i+1)} x_1^{(i, i + 1)}$ and $y_0^{(i, i+1)} y_1^{(i, i+1)}$; we also subdivide $x^{(k, k-1)}_0 x^{(k, k-1)}_1$ and $y^{(k, k-1)}_0 y^{(k, k-1)}_1$. We call the new subdivision nodes (in either case) $x^i_{\specialindex{0}}$ and $y^i_{\specialindex{0}}$, for $i \in [k]$, respectively.
	 
	For each $1 \le i < j \le k$, we subdivide the edges $x^{(i, j)}_p x^{(i, j)}_{p+1}$ and $y^{(i, j)}_p y^{(i, j)}_{p+1}$ and denote the new subdivision vertices by $x^{(i, j)}_{\specialindex{p}}$ and $y^{(i, j)}_{\specialindex{p}}$, respectively.	
	We illustrate this subdivision in Figure~\ref{fig:subdivision}.
	
	\begin{figure}
		\centering
		\includegraphics[width=.6\textwidth]{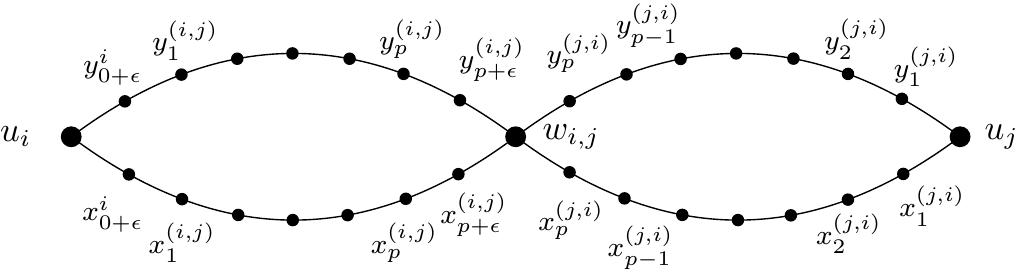}
		\caption{Illustration of the subdivision for a pair $1 \le i < j \le k$, assuming $j = i+1 < k$. For $j \neq i+1$, the vertices $x^i_{0+\epsilon}$ and $y^i_{0+\epsilon}$ do not exist.}
		\label{fig:subdivision}	
	\end{figure}
	
	We now construct the $H$-graph $G''$ by defining its $H$-representation $\cM = \{M_v\}_{v \in V(G')}$ where each $M_v$ is a connected subset of $V(H')$. (Recall that $G$ denotes the graph of the {\sc Multicolored Clique} instance.)
	\begin{enumerate}
		\item  For each $i \in [k]$, construct vertices $\alpha^i_x$ with model $M_{\alpha^i_x} \defeq \{x^i_{\specialindex{0}}\}$ and $\alpha^i_y$ with model $M_{\alpha^i_y} \defeq \{y^i_{\specialindex{0}}\}$.	
		\item For each $i \in [k]$ and $s \in [p]$, construct a vertex $z_s^i$ with model
			\begin{align*}
				M_{z_s^i} \defeq \left\lbrace x^i_{\specialindex{0}}, y^i_{\specialindex{0}}\right\rbrace \cup \bigcup\nolimits_{j \in [k], j \neq i}\left(\left\lbrace x_0^{(i, j)}, \ldots, x_{s - 1}^{(i, j)}\right\rbrace \cup \left\lbrace y_0^{(i, j)}, \ldots, y_{p-s}^{(i, j)}\right\rbrace \right).
			\end{align*}
		\item For each $1 \le i < j \le k$, construct a vertex $\alpha^{(i, j)}_x$ with model $M_{\alpha^{(i, j)}_x} \defeq \left\lbrace x^{(i, j)}_{\specialindex{p}}\right\rbrace$ and a vertex $\alpha^{(i, j)}_y$ with model $M_{\alpha^{(i, j)}_y} \defeq \left\lbrace y^{(i, j)}_{\specialindex{p}}\right\rbrace$.
		\item For each edge $v_s^i v_t^j \in E(G)$ for $s, t \in [p]$ and $1 \le i < j \le k$, construct a vertex $r_{s, t}^{(i, j)}$ with model
			\begin{align*}
				M_{r_{s, t}^{(i, j)}} \defeq &\left\lbrace x^{(i, j)}_{\specialindex{p}}, y^{(i, j)}_{\specialindex{p}} \right\rbrace
							\cup \left\lbrace x_s^{(i, j)},\ldots, x_{p+1}^{(i, j)}\right\rbrace 
							\cup \left\lbrace y_{p-s+1}^{(i, j)},\ldots, y_{p+1}^{(i, j)}\right\rbrace \\
							\cup &\left\lbrace x_t^{(j, i)}, \ldots, x_{p+1}^{(j, i)} \right\rbrace
							\cup \left\lbrace y_{p-t+1}^{(j, i)},\ldots, y_{p+1}^{(j, i)}\right\rbrace.
			\end{align*}	
	\end{enumerate}
	
	Throughout the following, for $i \in [k]$ and $1 \le i < j \le k$, respectively, we use the notation 
		\begin{align*}
			Z(i) \defeq \bigcup\nolimits_{s \in [p]}\left\lbrace z^i_s \right\rbrace  \mbox{ and } R(i, j) \defeq \bigcup\nolimits_{\substack{v^i_s v^j_t \in E(G), \\  s, t \in [p]}} \left\lbrace r_{s, t}^{(i, j)}\right\rbrace
		\end{align*}
	and we let $Z_{+\alpha}(i) \defeq Z(i) \cup \{\alpha^i_x, \alpha^i_y\}$ and $R_{+\alpha}(i, j) \defeq R(i, j) \cup \{\alpha^{(i, j)}_x, \alpha^{(i, j)}_y\}$.	
	We furthermore define
	\begin{align*}
		A \defeq \bigcup\nolimits_{i \in [k]} \left\lbrace \alpha^i_x, \alpha^i_y \right\rbrace \cup \bigcup\nolimits_{1 \le i < j \le k} \left\lbrace \alpha^{(i, j)}_x, \alpha^{(i, j)}_y \right\rbrace.
	\end{align*}	
	
	We obtain the graph $G'$ of the {\sc Maximum Induced Forest} instance by taking the graph $G''$ and adding to it a vertex $\beta$ which is adjacent to all vertices in $V(G'') \setminus A$. We illustrate this construction in Figure~\ref{fig:gprime}.
	
	\begin{figure}
		\centering
		\includegraphics[width=.5\textwidth]{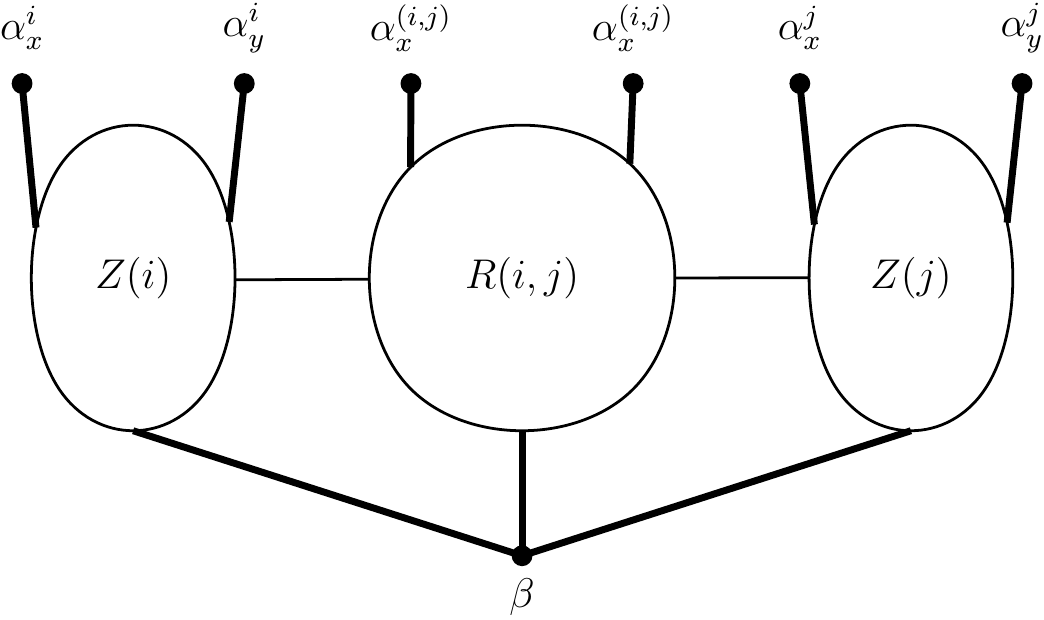}
		\caption{Illustration of a part of $G'$, where $1 \le i < j \le k$. Bold edges imply that all possible edges between the corresponding (sets of) vertices are present.}
		\label{fig:gprime}	
	\end{figure}
	
	We now show that the linear mim-width of $G'$ remains bounded by a function of $k$.\footnote{In fact, we will later show that $G'$ is a $K$-graph for some $K \supseteq H$.}
	\begin{claim}
		$G'$ has linear mim-width at most $4k(k-1)+1$ and a linear branch decomposition of said width can be computed in polynomial time.
	\end{claim}	
	\begin{clproof}
		By \cite[Theorem 1]{FGR17}, $G''$ has linear mim-width at most $2\cardd{H} = 4k(k-1)$. Given a linear branch decomposition of $G''$ we can add a new node to the branch decomposition in any place such that it stays linear and letting the new node be mapped to $\beta$. The resulting branch decomposition is a linear branch decomposition of $G'$ with the mim-value in each cut increased by at most $1$. 
		
		By \cite[Theorem 1]{FGR17} and the construction of the $H$-representation of $G''$ described above, this decomposition can be computed in polynomial time.
	\end{clproof}
	
	We now observe some crucial properties of the above construction.
	\begin{observation}[Claim 7 in~\cite{FGR17}]\label{obs:claim:7}
		For every $1 \le i < j \le k$, a vertex $z_h^i \in V(G')$ (a vertex $z_h^j \in V(G')$) is {\em not} adjacent to a vertex $r^{(i, j)}_{s, t}$ corresponding to the edge $v^i_s v^j_t \in E(G)$ if and only if $h = s$ ($h = t$, respectively).
	\end{observation}	
		
	\begin{observation}\label{obs:A}
		~
		\begin{enumerate}[(i)]
			\item For every $i \in [k]$, $N(\alpha^i_x) = Z(i) = N(\alpha^i_y)$.\label{obs:A:alpha:z}
			\item For every $1 \le i < j \le k$, $N(\alpha^{(i, j)}_x) = R(i, j) = N(\alpha^{(i, j)}_y)$.\label{obs:A:alpha:r}
			\item $A$ is an independent set in $G'$ of size $2k + 2\cdot\binom{k}{2}$.\label{obs:A:IS}
			\item For $i \in [k]$, $Z(i)$ induces a clique in $G'$ and for $1 \le i < j \le k$, $R(i, j)$ induces a clique in $G'$.\label{obs:A:cliques}
		\end{enumerate}
	\end{observation}	
	
	We are now ready to prove the correctness of the reduction. In particular we will show that $G$ has a multicolored clique if and only if $G'$ has an induced forest of size $k' \defeq 3k + 3\binom{k}{2} + 1$.
	\begin{claim}\label{claim:correctness:forward}
		If $G$ has a multicolored clique on vertex set $\left\lbrace v^1_{h_1}, \ldots, v^k_{h_k}\right\rbrace$, then $G'$ has an induced forest of size $k' = 3k + 3 \cdot\binom{k}{2} + 1$.
	\end{claim}
	\begin{clproof}
		Using Observation~\ref{obs:claim:7}, one can easily verify that the set
		\begin{align}
			I \defeq \left\lbrace z^1_{h_1}, \ldots, z^k_{h_k}\right\rbrace \cup \left\lbrace r^{(i, j)}_{h_i, h_j} \mid 1 \le i < j \le k \right\rbrace \label{eq:independent:set}
		\end{align}
		is an independent set in $G'$. By Observation~\ref{obs:A}(\ref{obs:A:IS}) and the construction given above, we can conclude that $F \defeq I \cup A \cup \{\beta\}$ induces a forest in $G'$: $I$ and $A$ are both independent sets and $A \cup I$ induces a disjoint union of paths on three vertices, the middle vertices of which are contained in $I$. The only additional edges that are introduced are between $\beta$ and vertices in $I$, so $F$ induces a tree. Clearly, $\card{F} = \card{I} + \card{A} + \card{\{\beta\}} = k + \binom{k}{2} + 2k + 2\cdot\binom{k}{2} + 1 = k'$, proving the claim.
	\end{clproof}
	
	We now prove the backward direction of the correctness of the reduction. This will be done by a series of claims and observations narrowing down the shape of any induced forest on $k'$ vertices in $G'$. Eventually, we will be able conclude that any such induced forest contains an independent set of size $k + \binom{k}{2}$ of the shape (\ref{eq:independent:set}). We can then conclude that $G$ contains a multicolored clique by Observation~\ref{obs:claim:7}.
	
	The following is a direct consequence of Observation~\ref{obs:A}(\ref{obs:A:cliques}).
	
	\begin{observation}\label{claim:correctness:backward:count}
		Let $F$ be an induced forest in $G'$. Then, $V(F)$ contains
				\begin{enumerate}[(i)]
					\item at most $2$ vertices from $Z(i)$, where $i \in [k]$ and\label{claim:correctness:backward:count:Z}
					\item at most $2$ vertices from $R(i, j)$, where $1 \le i < j \le k$.\label{claim:correctness:backward:count:R}
				\end{enumerate}
	\end{observation}
	
	Next, we investigate the interaction of any induced forest with the sets $Z_{+\alpha}(i)$ and $R_{+\alpha}(i, j)$.
	
	\begin{claim}\label{claim:two:from:zi}
		Let $F$ be an induced forest in $G'$. If $V(F)$ contains two vertices from $Z(i)$, where $i \in [k]$ (from $R(i, j)$, where $1 \le i < j \le k$), then $V(F)$ cannot contain a vertex from $\{\alpha^i_x, \alpha^i_y\}$ (from $\{\alpha^{(i, j)}_x, \alpha^{(i ,j)}_y\}$, respectively).
	\end{claim}	
	\begin{clproof}
		Suppose $V(F)$ contains two vertices $a, b \in Z(i)$. We prove the claim for $\alpha^i_x$ and note that the same holds for $\alpha^i_y$. By Observation~\ref{obs:A}(\ref{obs:A:cliques}), $a$ and $b$ are adjacent and $\alpha^i_x$ is adjacent to both $a$ and $b$ by Observation~\ref{obs:A}(\ref{obs:A:alpha:z}). Hence, $\{\alpha^i_x, a, b\}$ induces a $3$-cycle in $G'$. 
		
		An analogous argument can be given for the second statement.
	\end{clproof}
	
	In the light of Observation~\ref{claim:correctness:backward:count} and Claim~\ref{claim:two:from:zi}, we make
	\begin{observation}
		Let $F$ be an induced forest in $G'$. If $V(F)$ contains three vertices from $Z_{+\alpha}(i)$ for some $i \in [k]$ (three vertices from $R_{+\alpha}(i, j)$, respectively), then this set of three vertices must include $\alpha^i_x$ and $\alpha^i_y$ (resp., $\alpha^{(i, j)}_x$ and $\alpha^{(i, j)}_y$).
	\end{observation}
	
	The previous observation implies that in $G'$, any induced forest on $k' = 3k + 3 \cdot \binom{k}{2} + 1$ has the following form.
	\begin{enumerate}[(I)]
		\item For each $i \in [k]$, $V(F) \cap Z_{+\alpha}(i) = \{\alpha^i_x, \alpha^i_y, z^i_s\}$, for some $s \in [p]$.\label{eq:induced:forest:z}
		\item For each $1 \le i < j \le k$, $V(F) \cap R_{+\alpha}(i, j) = \{\alpha^{(i, j)}_x, \alpha^{(i, j)}_y, r^{(i, j)}_{t, t'}\}$, for some $t, t' \in [p]$.\label{eq:induced:forest:r}
		\item $\beta \in V(F)$.\label{eq:induced:forest:beta}
	\end{enumerate}
	To conclude the proof, we argue that any such induced forest $F$ includes an independent set of size $k + \binom{k}{2}$ of the form (\ref{eq:independent:set}). 
	In particular, we use the following claim to establish the correctness of the reduction.
	
	\begin{claim}\label{claim:independent:set}
		Let $F$ be an induced forest in $G'$ on $k'$ vertices, $1 \le i < j \le k$ and $s_i, s_j, t_i, t_j \in [p]$. If $z^i_{s_i}, r^{(i, j)}_{t_i, t_j}, z^j_{s_j} \in V(F)$, then $s_i = t_i$ and $s_j = t_j$.
	\end{claim}
	\begin{clproof}
		Suppose not and assume wlog.\ that $s_i \neq t_i$. Then, $\left\lbrace\beta, z^i_{s_i}, r^{(i, j)}_{s_t, t_i}\right\rbrace$ induces a $3$-cycle in $G'$: We have that $\beta \in V(F)$ by (\ref{eq:induced:forest:beta}), and by construction $\beta$ is adjacent to all vertices in $Z(i)$ and $R(i, j)$. By Observation~\ref{obs:claim:7} and the assumption that $s_i \neq t_i$, we have that $z^i_{s_i} r^{(i, j)}_{t_i, t_j} \in E(G')$.
	\end{clproof}
	
	Since by (\ref{eq:induced:forest:z}) and (\ref{eq:induced:forest:r}), any induced forest on $k'$ vertices contains precisely one vertex from each $Z(i)$ (for $i \in [k]$) and $R(i, j)$ (for $1 \le i < j \le k$), we can conclude together with Claim~\ref{claim:independent:set} that $V(F)$ contains an independent set 
	\begin{align*}
		\left\lbrace z^1_{s_1}, \ldots, z^k_{s_k}\right\rbrace \cup \left\lbrace r^{(i, j)}_{s_i, s_j} \mid 1 \le i < j \le k \right\rbrace
	\end{align*}
	which by Observation~\ref{obs:claim:7} implies that $G$ has a clique on vertex set $\left\lbrace v^1_{s_1}, \ldots, v^k_{s_k}\right\rbrace$.
\end{proof}

	Since a graph on $n$ vertices has an induced forest of size $k$ if and only if it has a feedback vertex set of size $n-k$, we have the following consequence of Theorem~\ref{thm:main}.

\begin{corollary}
	{\sc Feedback Vertex Set} is $\W[1]$-hard parameterized by linear mim-width, even if a linear branch decomposition of bounded mim-width is given.
\end{corollary}

We now show additionally that the above reduction can easily be modified to prove $\W[1]$-hardness for {\sc Maximum Induced Forest} and {\sc Feedback Vertex Set} on $H$-graphs when the parameter includes $\cardd{H}$. In particular, we show the following (using the notation from the proof of Theorem~\ref{thm:main}.)

\begin{proposition}\label{prop:K:graph}
	The graph $G'$ is a $K$-graph for some $K \supseteq H$ with $\card{K} = 3\cdot \card{H}$ and $\cardd{K} = \cardd{H} + 2\cdot\card{H}$.
\end{proposition}
\begin{proof}
	The graph $K$ is obtained from $H$ in the following way and is shown in Figure~\ref{fig:k}.
	\begin{enumerate}
		\item For each $i \in [k]$, add to $H$ two neighbors $\pi_i^x$ and $\pi^y_i$ of $u_i$.
		\item For each $1 \le i < j \le k$, add to $H$ two neighbors $\pi_{(i ,j)}^x$ and $\pi_{(i, j)}^y$ of $w_{(i, j)}$.
	\end{enumerate}
	\begin{figure}
	\centering
	\includegraphics[width=.35\textwidth]{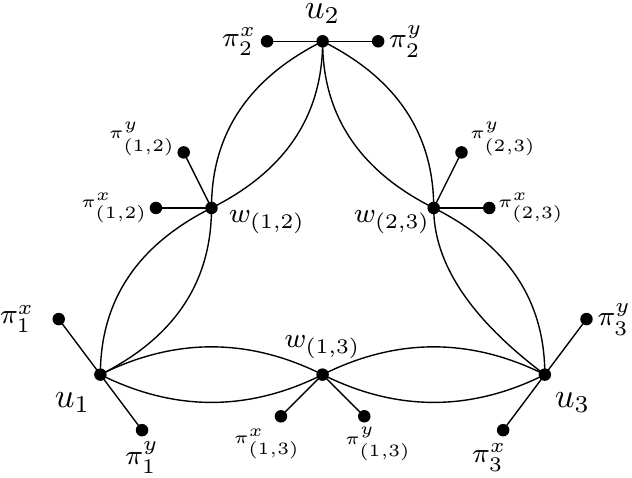}
	\caption{Illustration of the graph $K$ for $k = 3$.}
	\label{fig:k}
\end{figure}
	We let $\Pi \defeq \bigcup\nolimits_{i \in [k]} \{\pi_i^x, \pi_i^y\} \cup\bigcup\nolimits_{1 \le i < j \le k}\{\pi_{(i, j)}^x, \pi_{(i, j)}^y\}$.
	The subdivision $K'$ of $K$ is obtained from subdividing each each edge of $K[V(K) \setminus \Pi]$ $p$ times. (Note that this is the same subdivision done by Fomin et al.~\cite{FGR17}.) The graph $G'$ is now constructed similarly to the construction given in the previous proof, except that we do not have the vertices $x^{\cdot}_{\specialindex{0}}$, $y^{\cdot}_{\specialindex{0}}$, $x^{\cdot}_{\specialindex{p}}$ and $y^{\cdot}_{\specialindex{p}}$ in $K$ and hence in the models of the $K$-representation. For $i \in [k]$, the model of vertex $\alpha^i_x$ becomes $\{\pi_i^x\}$ and the model of $\alpha^i_y$ becomes $\{\pi_i^y\}$. For $1 \le i < j \le k$, the model of $\alpha^{(i, j)}_x$ becomes $\{\pi_{(i, j)}^x\}$ and the model for $\alpha^{(i, j)}_y$ becomes $\{\pi_{(i, j)}^y\}$. Furthermore, the model of each vertex $z^i_s$ includes $\{\pi_i^x, \pi_i^y\}$ and the model of each $r^{(i, j)}_{s, t}$ includes the nodes $\pi_{(i, j)}^x$ and $\pi_{(i, j)}^y$.
	We can now represent the vertex $\beta$ with model $V(K) \setminus \Pi$.
	
	It is straightforward to verify that this procedure gives a $K$-representation of $G'$.
\end{proof}

By Proposition~\ref{prop:K:graph} we have the following consequence of the proof of Theorem~\ref{thm:main}.
\begin{corollary}
	{\sc Maximum Induced Forest} on $H$-graphs is $\W[1]$-hard when parameterized by $k + \cardd{H}$ and {\sc Feedback Vertex Set} on $H$-graphs is $\W[1]$-hard when parameterized by $\cardd{H}$. In both cases, the hardness even holds when an $H$-representation of the input graph is given.
\end{corollary}

\bibliography{fvsmim}
\end{document}